\documentclass[letterpaper, 10 pt, conference, final]{ieeeconf}

\IEEEoverridecommandlockouts

\usepackage[T1]{fontenc}
\usepackage[utf8]{inputenc}

\usepackage{booktabs,threeparttable}
\usepackage{tabularx,ragged2e,booktabs,caption, subcaption}
\usepackage{multirow}
\usepackage{multicol}
\usepackage{siunitx} 

\usepackage{amssymb,amsmath}
\usepackage{array, makecell}
\usepackage{units}
\usepackage{csquotes}
\usepackage{xspace}
\usepackage{textcomp}

\usepackage[ampersand]{easylist}

\usepackage[utf8]{inputenc}
\usepackage{amsmath,amssymb,amsfonts}
\usepackage{algorithmic}
\usepackage{graphicx}
\usepackage{textcomp}
\usepackage{xcolor}
\usepackage{lpic}
\usepackage{relsize}
\usepackage{etoolbox}
\usepackage[acronym]{glossaries}

\usepackage{color}
\definecolor{plecs1}{rgb}{0,.7969,0}
\definecolor{plecs2}{rgb}{1,0,0}
\definecolor{plecs3}{rgb}{0,0,1}
\definecolor{plecs4}{rgb}{1,.7969,0}
\definecolor{plecs5}{rgb}{1,0,1}
\definecolor{plecs6}{rgb}{0,1,1}

\definecolor{vgRed}{RGB}{193, 48, 24}
\definecolor{vgOrange}{RGB}{243, 111, 19}
\definecolor{vgYellow}{RGB}{235, 203, 56}
\definecolor{vgGreen}{RGB}{162, 185, 105}
\definecolor{vgLightBlue}{RGB}{13, 149, 188}
\definecolor{vgDarkBlue}{RGB}{6, 56, 81}

\usepackage{acronym}
\acrodef{mg}[MG]{microgrid}
\acrodef{mpc}[MPC]{model predictive control}
\acrodef{dmpc}[DMPC]{distributed model predictive control}
\acrodef{ems}[EMS]{energy management system}
\acrodef{res}[RES]{renewable energy sources}
\acrodef{al}[AL]{augmented Lagrangian}
\acrodef{dg}[DU]{distributed unit}
\acrodef{dl}[DL]{demand load}
\acrodef{pcc}[PCC]{point of common coupling}
\acrodef{admm}[ADMM]{alternating direction method of multipliers}
\acrodef{dd}[DD]{dual decomposition}
\acrodef{ac}[AC]{alternating current} \acused{ac}
\acrodef{dc}[DC]{direct current} \acused{dc}
\acrodef{facts}[FACTS]{flexible \ac{ac} transmission system}
\acrodefplural{facts}[FACTS]{flexible \ac{ac} transmission systems}
\acrodef{mip}[MIP]{mixed integer program}
\acrodef{miqp}[MIQP]{mixed integer quadratic program}
\acrodef{ocp}[OCP]{optimal control problem}
\acrodef{opf}[OPF]{optimal power flow}
\acrodef{saa}[SAA]{sample average approximation}
\acrodef{cf}[CF]{communication failure}
\acrodef{pv}[PV]{photovoltaic}
\definecolor{philipp}{RGB}{205, 102, 000}
\definecolor{christian}{RGB}{25, 122, 000}

\def\R{\mathbb{R}}

\def\N{\mathbb{N}}

\newcommand{\diag}[1]{\operatorname{diag}(#1)}

\newcommand{\hh}{\nobreak\hspace{0pt}-\hspace{0pt}}
\newcommand{\T}{^\mathsf{T}}

\newcommand{\RR}{\mathbb{R}}
\newcommand{\Rp}{\mathbb{R}_{>0}}
\newcommand{\Rpz}{\mathbb{R}_{\geq0}}
\newcommand{\Rnz}{\mathbb{R}_{\leq0}}
\newcommand{\Npz}{\mathbb{N}_{\geq0}}

\newcommand{\tmin}{\text{min}}
\newcommand{\tmax}{\text{max}}

\newcommand{\tst}{\text{s}}
\newcommand{\tth}{\text{t}}
\newcommand{\trs}{\text{r}}
\newcommand{\tll}{l}
\newcommand{\tel}{\text{el}}

\newcommand{\ok}{(k)}

\newcommand{\okk}{(k+1)}
\newcommand{\okm}{(k-1)}
\newcommand{\okj}{(k+j)}

\newcommand{\Ts}{\mathrm{T}_\text{s}}
\newcommand{\OneT}{\mathbf{1}^{\T}}

\newcommand{\pmin}{p^\tmin}
\newcommand{\pmax}{p^\tmax}

\newcommand{\pth}{p_\tth}
\newcommand{\pst}{p_\tst}
\newcommand{\prs}{p_\trs}

\newcommand{\uth}{u_\tth}
\newcommand{\ust}{u_\tst}
\newcommand{\urs}{u_\trs}
\newcommand{\wrs}{w_\trs}
\newcommand{\wll}{w_\tll}

\newcommand{\psti}{p_{\tst,i}}

\newcommand{\ptk}{\pth\ok}

\newcommand{\psk}{\pst\ok}

\newcommand{\prk}{p_\trs\ok}

\newcommand{\hatw}{\hat{w}}

\newcommand{\psik}{p_{\tst,i}\ok}

\newcommand{\xk}{x\ok}

\newcommand{\xsmin}{x^\tmin}
\newcommand{\xsmax}{x^\tmax}

\newcommand{\deltath}{\delta_\tth}

\newcommand{\deltatk}{\deltath\ok}

\newcommand{\com}{\operatorname{com}}

\newcommand*\rfrac[2]{{}^{#1}\!/_{#2}}

\usepackage{algorithm}

\usepackage[breaklinks=true]{hyperref}
\usepackage{cleveref}
\newtheorem{problem}{Problem}
\newtheorem{proposition}{Proposition}
\newtheorem{theorem}{Theorem}[section]
\newtheorem{remark}[theorem]{Remark}
\renewcommand{\proof}[1][]{\ifthenelse{\equal{#1}{}}{\noindent\hspace{2em}{\itshape Proof: }}{\noindent\hspace{2em}{\itshape Proof #1: }}\@\xspace}

\usepackage{changes}

\begin{document}

\title{Fallback Strategies in Operation Control of Microgrids with Communication Failures\\
}

\author{{I. Löser}, {A.\,K. Sampathirao}, {S. Hofmann} and {J. Raisch}
\thanks{I. Löser, A.\,K. Sampathirao and S. Hofmann are with the Control Systems Group, Technische Universit\"at Berlin, Germany,
\href{mailto:loeser@campus.tu-berlin.de}{\texttt{loeser@campus.tu-berlin.de}}, \href{mailto:sampathirao@control.tu-berlin.de}{\texttt{sampathirao@control.tu-berlin.de}}, \href{mailto:hofmann@control.tu-berlin.de}{\texttt{hofmann@control.tu-berlin.de}}.}
\thanks{J. Raisch is with the Control Systems Group, Technische Universit\"at Berlin, Germany and Max-Planck-Institut f\"ur Dynamik komplexer technischer Systeme, Magdeburg, Germany, \href{mailto:raisch@control.tu-berlin.de}{\texttt{raisch@control.tu-berlin.de}}.}
\thanks{The work was supported by the German Federal Ministry for Economic Affairs and Energy (BMWi), Project No. 0324024A. We also acknowledge funding by the DFG priority program 1-500047901EF.}
}

\maketitle


\begin{abstract}
This paper proposes a \ac{mpc}-based energy management system to address communication failures in an islanded \ac{mg}. The energy management system relies on a communication network to monitor and control power generation in the units. A communication failure to a unit inhibits the ability of the \acs{mpc} to change the power of the unit. However, this unit is electrically connected to the network and ignoring this aspect could have adverse effect in the operation of the \acl{mg}. Therefore, this paper considers an \ac{mpc} design that includes the electrical connectivity of the unit during communication failures. This paper also highlights the benefit of including the lower layer control behaviour of the \ac{mg} to withstand communication failures. The proposed approaches are validated with a case study.
\end{abstract}



\section{Introduction} \label{ch:01_introduction}

Reducing greenhouse emissions from the electrical sector leads to a worldwide increase in the installation of \ac{res}~\cite{Sec2018}.
\ac{res} are small-scale \acp{dg} such as \ac{pv} plants and wind
turbines, characterised by intermittent generation.
The volatility of \ac{res} has created new challenges for their integration into the power system.
In this context, \acfp{mg} are considered as an attractive solution to address the volatility of \ac{res}.
An \ac{mg} is a self-reliant, small-scale power system defined in a specified electrical region that manages its local load demands with its local \acp{dg}~\cite{DOE2012}.

A typical \ac{mg} comprises \ac{res}, storage units like battery plants and thermal units like diesel generators. The \ac{mg} has an \ac{ems} to facilitate its economical and reliable operation. This \ac{ems} coordinates the charging/discharging schedule of the storage units and the switching of the thermal units, taking into account the \ac{res} generation and load demand. Model predictive control (\ac{mpc}) is a popular control strategy used in the design of \ac{ems}~\cite{ParRikGli2014, SolDinJor2017, HSR+2019}.

Coordination by the \ac{ems} requires a communication infrastructure that allows to monitor and control the units~\cite{gao2012survey}. The communication network is prone to various disruptions like sensor and software failures, problems in the communication interface or human errors that lead to disconnecting the network interface. These failures result in an interruption of the communication between the corresponding unit and the \ac{ems}. In the following, the loss of communication is referred to as \ac{cf}.

In recent years, the impact of communication loss on power systems has been recognised, e.g., \cite{FalFuWu2012,SiqMoh2016}.
In~\cite{AmiFotSha2012}, the impact of communication loss is measured through Monte Carlo simulation.
In~\cite{MouAkaDeb+2018}, the authors considered a combined model for the power and communication network to identify the critical links whose failure can lead to maximal disruption in the network.
In~\cite{FalKarFu2013}, the authors stated that loss of communication results in loss of control and quantified the resulting loss as the difference between the minimum operation cost of the power system with and without \ac{cf}.
The main limitation of these works is that they neither discuss operation strategies during a \ac{cf} nor do they account for \ac{res} generation.

In general, an \ac{ems} is equipped with a control mode that is concerned with failures, e.g.,~\cite{ZidKhaAbd2017}.
In~\cite{ProEnrFlo2015}, a fault\hh tolerant \ac{mpc}\hh based \ac{ems} is proposed to ensure proper amount of energy in the storage devices such that the demand can be covered. However, the existing methods in literature do not distinguish between \acp{cf} and electrical failures in the system. Thereby, there is a lack of targeted strategies to address \ac{cf}. A failure in the communication line does not necessarily imply an electrical disconnection of the corresponding unit. Particularly in remote \acp{mg}, disconnecting the units and rescheduling the power generation with the remaining units could be costly. A recent paper that addresses controller design with \acp{cf} is \cite{zhang2018joint}. A decentralised control that can achieve near-optimal performance without any communication is proposed. However, this work does not take into account the storage dynamics and uncertainties of the \ac{res}.


There is ample research addressing control design with unreliable communication networks (see, e.g., \cite{ParMahDin+2018}).
The authors of \cite{pena2008lyapunov} propose a Lyapunov-based \ac{mpc} that considers data losses in the network, with the actuators using the last received optimal trajectory in the event of \ac{cf}.
A similar strategy is considered at the actuator during loss of communication in \cite{GiuFraDom2015}.
These methods can be included in the \ac{mpc} for the \ac{mg}.
However, \acp{mg} employ a hierarchical control scheme (see, e.g.,~\cite{VanVasKoo+2013}) and not taking this into account could limit the benefits from the above approaches.

This paper focuses on developing a fallback strategy for \acp{cf} that considers the hierarchical control scheme of the \ac{mg} in an \ac{mpc}\hh based \ac{ems}.
In particular, it concentrates on \acp{cf} to storage units because scheduling their charging/discharging is central for the economical operation and integration of \ac{res}. Furthermore, the storage units have predefined storage capacities which must not be violated.

The remainder of this paper is organized as follows.
Section \ref{sec:system:model} introduces a combined model of the \ac{mg} and the communication network.
In Section~\ref{sec:control:problem}, the control objectives and the \ac{mpc} formulation for the \ac{mg} are proposed. Section~\ref{sec:fall:back:strategy} describes the \ac{mpc} problem with the fallback strategy when the \ac{cf} occurs at a storage unit.
Finally, we present a case study in Section~\ref{sec:caseStudy}. 

\subsection{Notation and mathematical preliminaries}
\label{sec:preliminaries}
The set of real numbers is denoted by $\R$, the set of nonpositive real numbers by $\R_{\leq 0}$ and the set of positive real numbers by $\R_{>0}$.
The set of positive integers is denoted by $\mathbb{N}$ and
the set of the first $n$ positive integers by $\N_{n} = \{1, 2, \ldots, n\}$. The set of integers is denoted by $\mathbb{Z}$. Now $\mathbb{Z}_{[a, b]}, a \leq b$ is the set $\{a, a+1, \ldots, b\}$ and $\mathbb{Z}_{[a, b]}, a > b$ is the empty set.
The cardinality of a set $\mathbb{A}$ is $|\mathbb{A}|$.
For a vector or matrix $x$, the transpose is denoted by $x\T$.
For a vector $x \in \RR^n$ with elements $x_i$ we define $\mathbf{1}\T x = \sum_{i = 1}^{n} x_{i}$.
For $x \in \RR$ and $\delta \in \{0,1\}$, we define $y=\delta\wedge x$ as: $y=x$ if $\delta=1$ and $y=0$ if $\delta=0$.
When $x$ and $\delta$ are vectors, this operation is performed element\hh wise. For a given vector $x \in \RR^{n}$, the diagonal matrix generated  from it is denoted as $\diag{x} \in \RR^{n \times n}$. For a set $A$, $A \backslash j = \{ i \in A | i \neq j\}$.


\section{System model}\label{sec:system:model}
This section introduces the discrete-time model of an islanded \ac{mg} that explicitly defines the behaviour of the units with and without \ac{cf}.
The communication between the \ac{ems} and the units of the \ac{mg} is two-way and it is assumed that both the \ac{ems} and the units can detect when the communication link fails.
The sampling time of the \ac{mpc} is several minutes and therefore this model is the steady-state model. It is based on our previous work in~\cite{HSR+2019}.

\subsection{\ac{mg} description} \label{sec:mgModel}
An \ac{mg} is an electrical system composed of \acp{dg} and loads.
We assume that all \acp{dg} in an \ac{mg} are either thermal units, storage units or \ac{res}.
Let us denote the number of \acp{dg} and loads by $g$ and $l$.
Each unit is labelled as $\ac{dg}_{i} $ with $ i \in \N_g$.
Furthermore, we define the index set of the thermal units as
$T = \{ i \mid i \in \N_{g}, \ac{dg}_{i} $ is a thermal unit$\}$.
Similarly, we define $S$ as the index set of storage units and $R$ as the index set of \ac{res}.
In the same way, the loads are labelled as $\acs{dl}_{i}, i \in \N_l = L$.
Note that the number of thermal units, storage units and renewable units in an $\ac{mg}$ is $|T|$, $|S|$ and $|R|$ with $g = |T| + |S| + |R|$.

At a given time instance $k \in \Npz$, let us denote the uncertainty as $w\ok = [w_\trs\ok\T ~ w_\tll\ok\T]\T$, where $w_\trs\ok \in \Rpz^{|R|}$ is the available renewable infeed and $w_\tll\ok \in \Rnz^{|L|}$ is the load demand.
The power set-points provided by the \ac{ems} are denoted by $u\ok = [u_\tth\ok\T ~ u_\tst\ok\T ~ u_\trs\ok\T ]\T$, where $u_\tth\ok\in \Rpz^{|T|}$, $u_\tst\ok\in\RR^{|S|}$ and $u_\trs\ok\in \Rpz^{|R|}$ are the set\hh points of the thermal units, storage units and \acs{res} respectively.
Similarly, the power output from the units is denoted by $p\ok = [p_\tth\ok\T ~ p_\tst\ok\T ~ p_\trs\ok\T]\T$, where $p_\tth\ok \in \Rpz^{|T|}, p_\tst\ok \in \RR^{|S|}$ and $p_\trs\ok \in \Rpz^{|R|}$.
Each thermal unit can be switched on or off, which is represented by the binary variable $\delta_\tth\ok \in \{0, 1\}^{|T|}$.
In storage units, the energy level is denoted by $x\ok \in \Rpz^{|S|}$.

\subsubsection{Communication network model}
The communication status of a unit is defined by the binary variable $\zeta \in \{0, 1\}$, where $\zeta = 1$ represents active communication and $\zeta = 0$ denotes a \ac{cf}. When $\zeta = 0$, the \ac{ems} can neither provide set-points to the unit nor can it receive new measurements from the unit. In such cases, the unit either disconnects from the \ac{mg} or uses a default power set-point~\cite{KosCosBin2013}.
The power set\hh point of a unit with the communication status $\zeta$ is given as
\begin{equation}\label{eq:communication}
\com(u, d, \zeta) := (1 - \zeta)d  + \zeta u
\end{equation}
where $d$ is the default power and $u$ is the power set\hh point provided by the \ac{mpc}. When $p$, $d$ and $\zeta$ are vectors, then $\com(\cdot, \cdot, \cdot)$ is element-wise.

At time step $k$, let us denote the communication status of all units as
$\zeta\ok = [\zeta_{\tth}\ok\T \zeta_{\tst}\ok\T \zeta_{\trs}\ok\T]\T$, where $\zeta_{\tth}\ok \in \{0, 1\}^{|T|}$, $\zeta_{\tth}\ok \in \{0, 1\}^{|S|}$ and $\zeta_{\tth}\ok \in \{0, 1\}^{|R|}$ are the
communication status of the thermal, storage and renewable units,  respectively.

In the \ac{res}, the power output $\prs\ok$ depends on the available power $\wrs\ok$ and is given as
\begin{subequations}\label{eq:unit:power}
\begin{equation}\label{eq:renewable:power}
  \prs\ok = \min(\com(\urs\ok, d_\trs\ok, \zeta_\trs\ok), \wrs\ok).
\end{equation}

As a part of the hierarchical control scheme of the \ac{mg}, there are local controllers at the storage and thermal units.
The steady-state power of a storage unit is a sum of two parts: the power corresponding to the power set\hh point provided by the \ac{ems} and the power injected by the local controller~\cite{KriSchRai2019}.
This can be represented as
\begin{equation}\label{eq:storage:power}
  \pst\ok = \com(\ust\ok, d_\tst\ok, \zeta_\tst\ok) + \chi_{s} \rho\ok,
\end{equation}
where $\chi_{\tst} \in \Rp^{|S|}$ is a positive gain vector and $\rho\ok \in \R$ is proportional to frequency deviation in \ac{mg}.
Note that the power injected by the local controller is independent of the communication status.

The power of the thermal units can be written similarly to the storage units. However, the thermal unit can generate power only when it is switched on.
Therefore, the power is given as
\begin{align}\label{eq:thermal:power}
  \ptk =& \com(\deltatk, \delta_d\ok, \zeta_\tth\ok) \wedge \chi_{t} \rho\ok \nonumber\\
  & + \com(\uth\ok, d_\tth\ok, \zeta_\tth\ok),
\end{align}
where $\chi_{\tth} \in \Rp^{|T|}$ is a positive gain and $\delta_d\ok \in \{0,1\}^{|T|}$ is the default switch state.
Here, the default power $d_{\tth}\ok$ is zero when the default switch state $\delta_d\ok = 0$. In words, in case of a \ac{cf}, the default switch state $\delta_d\ok$ determines if the local controller is active or inactive.
\end{subequations}

\subsubsection{Constraints on the units}
In an \ac{mg}, the power generated should match the load demand. At time step $k \in \N$, this condition is represented by the power balance equation
\begin{equation}\label{eq:powerBalance}
\OneT p_\tth\ok + \OneT p_\tst\ok + \OneT p_\trs\ok + \OneT w_\tll
\ok = 0.
\end{equation}
In the above equation~\eqref{eq:powerBalance}, the load demand $w_\tll \ok$ and the renewable power $p_\trs\ok$ are uncertain.
The $\rho\ok$ in the steady-state power output of the storage unit~\eqref{eq:storage:power} and conventional unit \eqref{eq:thermal:power} ensures that these fluctuations are compensated.

The constraints on the power set\hh points and the power of the units are defined as:
\begin{subequations}\label{eq:power:limits}
  \begin{alignat}{2}
    \pmin_{\tst} &\leq \psk & &\leq\pmax_{\tst},\label{eq:st:pminmax}\\
    \pmin_{\trs} &\leq \prk & &\leq\pmax_{\trs},\label{eq:rs:pminmax}\\
    \deltatk\wedge\pmin_{\tth} &\leq \ptk & &\leq\deltatk \wedge \pmax_{\tth},\label{eq:th:pminmax}
  \end{alignat}
with $\pmin_{\tst} \in \RR^{|S|}, \pmin_{\trs} \in \Rpz^{|R|}, \pmin_{\tth} \in \Rp^{|T|}$ and $\pmax_{\tst} \in \RR^{|S|}, \pmax_{\trs} \in \Rp^{|R|}, \pmax_{\tth} \in \Rp^{|T|}$.
The default power set\hh point $d\ok$ and power set\hh points $u\ok$ in~\eqref{eq:unit:power} must satisfy the above constraints.

The energy capacities of the storage units are represented by
  \begin{alignat}{2}
    \xsmin &\leq \xk & &\leq\xsmax \label{eq:xminmax},
  \end{alignat}
with $\xsmin\in\Rpz^{|S|}$, $\xsmax\in\Rp^{|S|}$.

The dynamic model of the storage is given as
\begin{align}\label{eq:storage:dynamics}
  x\okk = x\ok - \diag{b_{\tst}}\psk
\end{align}
where $b_{\tst} \in \Rp^{|S|}$ is the charging efficiency vector.
The charging efficiency depends on whether the storage is being charged or discharged and can be given as
\begin{align}\label{eq:storage:matB}
  b_{\tst,i} = \begin{cases}
  \Ts \eta_{\tst ,i} & \psik < 0\\
  \Ts /\eta_{\tst ,i} & \text{otherwise}
  \end{cases}
\end{align}
with the efficiency $\eta_{\tst, i} \in (0, 1]$ and the sampling time $\Ts \in \Rp$.

The units in the \ac{mg} are connected by an electrical network. Let us denote the number of power lines in this network as $m$. The power flow in these lines results in losses.
However, we do not consider these losses in the network as they are small compared to the uncertainty in the \ac{mg}.
Then, the constraints of the electrical network can be given as
\begin{align}
p_{\tel}^{\tmin} \leq H p\ok \leq p_{\tel}^{\tmax}
\end{align}
where $H \in \RR^{m \times (g+l)}$ depends on the topology of the network (see, e.g.,~\cite{StoJarAls2009}) and $p_{\tel}^{\tmin} \in \RR^{m}$, $p_{\tel}^{\tmax} \in \RR^{m}$ are the line limits.
\end{subequations}

Note that although the losses are neglected, they can easily modelled either as localised losses or as quadratic losses with semidefinite relaxations~\cite{EldNeiCas2018}.
These and other extensions can be easily included in the \ac{mpc} formulation.


\section{Control problem of the \ac{mg}}\label{sec:control:problem}
This section explains the \ac{mpc} problem formulation for the \ac{mg}. 

\subsection{Operation cost}
\ac{mpc} determines the sequence of power set\hh points that minimises an objective function encompasses economical and safety objectives.

In case of thermal units, operation cost corresponds to fuel cost and switching cost. The fuel cost can be approximated by a quadratic function (see, e.g., \cite{JMK2012}). The switching cost is included to penalise switching on and off.
The cost for the thermal unit is then
\begin{subequations}
  \begin{equation} \label{eq:costGen}
    \ell_{\tth}(p_{\tth}, \delta_{\tth}, \delta_{\tth}^{s}) =  a_{\tth}\T\Delta\delta_{\tth} + a_{\tth, 1}\T \delta_{\tth} + a_{\tth, 2}\T p_{\tth} + p_{\tth}\T \diag{a_{\tth, 3}} p_{\tth},
  \end{equation}
  with weights $a_{\tth} \in \Rpz^{|T|}, a_{\tth, 1}, a_{\tth, 2}, a_{\tth, 3} \in \R_{> 0}^{|T|}$ and $\Delta\delta_{\tth} = |\delta_{\tth} - \delta_{\tth}^{s}|$ where $\delta_{\tth}^{s} = \delta_{\tth}(k - 1)$.

For batteries, we consider cost that penalises high power charging and discharging and energy levels above or below a threshold that can negatively impact the lifespan of the battery. Let us denote the thresholds as $[x_{\tst}^{\tmin}, x_{\tst}^{\tmax}]$.
Now the cost of the storage is
\begin{equation} \label{eq:costStorage}
  \ell_{\tst}(\pst, x) = \pst\T \diag{a_{\tst}} \pst + \Delta x\T \diag{a_{\tst, 1}} \Delta x,
\end{equation}
where $\Delta x = \max(x_{\tst}^{\tmin} - x, 0) +  \max(x - x_{\tst}^{\tmax}, 0)$, $a_{\tst} \in \Rpz^{|S|}$ and $a_{\tst, 1} \in \Rpz^{|S|}$.

Although \ac{res} do not have any fuel cost, the energy wasted from renewable units would potentially lead to conventional generation. Therefore, a cost incentivising its usage is included
\begin{equation}
\ell_{\trs}(\prs) = -a_{\trs}\T \prs ,
  \end{equation}
  where $a_{\trs} \in \Rpz^{|R|} $ is a weight.
  Finally, the overal operation cost of the \ac{mg} is then
  \begin{align}\label{eq:sumOfUnitCosts}
      \ell(p, \delta_{\tth}, x, \delta_{\tth}^s) = \ell_{\tth}(p_{\tth}, \delta_{\tth}, \delta_{\tth}^s) + \ell_{\trs}(\prs) + \ell_{\tst}(\pst, x)
  \end{align}
\end{subequations}

\subsection{\ac{mpc} problem}
An \ac{mpc} scheme uses the plant model to predict the future operation of the \ac{mg} and to calculates a sequence of power set\hh points that minimises the operation cost.
The \ac{mg} model presented in Section~\ref{sec:system:model} includes available renewable infeed $\wrs$, load demand $\wll$ and the communication status $\zeta$.

In a certainty equivalence \ac{mpc} design, the model uses forecasts of the available renewable infeed and the load, disregarding information on the quality of the forecasts(see, e.g., \cite{ParRikGli2014, SolDinJor2017}).
At time step $k$, let us denote the forecasts for the available renewable infeed at $k+j$ as $\hatw_{\trs}(k+j |k)$ and load demand $\hatw_{\tll}(k+j |k)$. It is assumed that the communication status observed at time step $k$ holds over the future, i.e.,
\begin{align}
\hat{\zeta}(k+j |k) = \zeta\ok, j > 0.\nonumber
\end{align}
Finally, the $\rho\ok$ to ensure power balance~\eqref{eq:powerBalance} is zero as this \ac{mpc} assumes exact information over the prediction horizon.

Let us denote the power set\hh points and the power output predicted by the \ac{mpc} by $u(k + j|k)$, $p(k + j|k), j > 0$ respectively. To simplify notation, we write this as $u(k + j)$, $ p(k + j)$ respectively. For a prediction horizon $h \in \N$, the decision variables are  $\mathbf{u} = [u\ok ~ \cdots ~ u(k + h)]$ and $\boldsymbol{\delta} = [\delta_{\tth}\ok ~ \cdots ~ \delta_{\tth}(k + h)]$ and the corresponding power is $\mathbf{p} = [p\ok ~ \cdots ~ p(k + h)]$.

The standard certainty equivalence \ac{mpc} problem as in \cite{ParRikGli2014, SolDinJor2017} can be formulated as
\begin{problem}[\ac{mpc} with \ac{cf}] \label{prob:simple:mpc}
  \begin{subequations}
    \[
      \min_{\mathbf{u}, \boldsymbol{\delta}} \ V(\mathbf{p}, \boldsymbol{\delta}),
    \]
    where \[
    V(\mathbf{p}, \boldsymbol{\delta}) = \textstyle\sum_{j = 0}^{h} \gamma^{j}\ell(p\okj, \deltath\okj, x\okj, \delta(k-1)))
    \]
    subject to \eqref{eq:unit:power}, \eqref{eq:powerBalance}, \eqref{eq:power:limits}, current energy levels $x\ok$ and previous switch statuses $\deltath\okm$.
    Furthermore the forecasts of the available \ac{res}, load demand and communication status are
    \begin{equation}\label{eq:problem:forecast}\begin{split}
        \wrs\okj = \hatw_{\trs}\okj, \\
        \wll\okj = \hatw_{\tll}\okj, \\
        \zeta\okj = \hat{\zeta}\okj.
      \end{split}
    \end{equation}
    In the above formulation, the power of the units~\eqref{eq:unit:power} depends on $\rho\ok$.
    However, standard certainty equivalence formulation does not include local control. Here, this can be realised with the constraint
    \begin{align}\label{eq:certain:equivalent}
      \rho\okj = 0,
    \end{align}
    for $j \in \{ 0, 1, \cdots, h \}$. Here $\gamma \in (0, 1]$ is a discount factor.
\end{subequations}
\end{problem}

In the above problem, when a \ac{cf} occurs in a unit $i$, i.e., $\zeta_{i}\ok  = 0$, the unit produces the default power $p_{i}\okj = d_{i}\okj$.
In \ac{mpc} design with unreliable communication, the default power set\hh points are obtained from the \ac{mpc} solution before the \ac{cf}~\cite{pena2008lyapunov,GiuFraDom2015}. At time step $k$, this is given as
\begin{equation}\label{eq:fallBack:stategy}
		d_i\okj = \begin{cases}
		u_i(k + j | k - c ), j \in \mathbb{Z}_{[0, h - c]}\\
		u_i( k - c + h | k - c), j\in \mathbb{Z}_{[h - c + 1, h]}
	\end{cases}
  \end{equation}
$ \forall i \in \N_{g}$, where $k - c$ is the last time instance before the \ac{cf} has occurred. Finally,  $u_i(\cdot | k - c )$ is the \ac{mpc} solution obtained from Problem~\ref{prob:simple:mpc} at $k - c$.



\section{Communication fall back strategies}\label{sec:fall:back:strategy}

The main limitations of Problem~\ref{prob:simple:mpc} are 1) knowledge of the current energy state of the battery depends on the communication network 2) this formulation does not include the local control behaviour of the battery and thermal units. Neglecting local control restricts the ability to modify the power output during \ac{cf}. In this section, we extend the \ac{mpc} formulation to include local control to address \acp{cf}. In particular, we formulate the problem for \acp{cf} in the storage units which is particulary challenging as the current energy levels are not directly available to the \ac{mpc}.

\subsection{Estimation of energy level of the battery}
In Problem~\ref{prob:simple:mpc}, when the communication to storage $i \in S$ has failed, then the corresponding $x_{i}$ is not available. However, this can be estimated from the default power set\hh points \eqref{eq:fallBack:stategy}, power output of the units~\eqref{eq:unit:power} and the power balance equation~\eqref{eq:powerBalance}.
In Problem~\eqref{prob:simple:mpc}, at time instance $k + 1$, the previous used forecasts are $\hatw_{\trs}\ok$ and $\hatw_{\tll}\ok$. The actual previous renewable infeed and load are $w_{\trs}\ok$ and $w_{\tll}\ok$.
Therefore, $\rho\ok$ that ensures power balance for the actual uncertainty is

\begin{subequations}\label{eq:estimation:storage}
	\begin{equation}\label{eq:rho:estimation}
		\rho\ok = \frac{(\OneT\Delta \wll \ok + \OneT\Delta \wrs \ok)}{\chi},
	\end{equation}
	where
	\begin{equation}\begin{split}
		\Delta \wll \ok =& \hatw_{\tll}\ok {- \wll\ok},\\
		\Delta \wrs \ok =& \min(\com(\urs\ok, d_\trs\ok, \zeta_\trs\ok), \hatw_{\trs}\ok)  \\
	 &{-\min(\com(\urs\ok, d_\trs\ok, \zeta_\trs\ok), \wrs\ok)},\\
	 \chi =& \OneT(\com(\deltatk, \delta_d\ok, \zeta_\tth\ok) \wedge \chi_{t}) + \OneT \chi_{\tst}.\nonumber
	\end{split}
\end{equation}
Now the power output of the battery with \ac{cf} is
\begin{align}
	\psti\ok = d_{\tst, i}\ok + \chi_{\tst, i} \rho \ok
\end{align}
and the resulting energy level is obtained by substituting this power in the storage dynamics given by Equation \eqref{eq:storage:dynamics} as
\begin{align}
	\hat{x}_{\tst, i}(k+1) = x_{\tst, i}\ok - b_{\tst, i}\psti\ok
\end{align}
\end{subequations}

\begin{remark}
The advantage of the above estimation method is that it works independently of the number of storage units with \ac{cf}.
If the losses are included in the network model,
$\rho\ok$ is obtained from solving the non-linear power equations.
\end{remark}

\begin{remark}
An alternative estimation method can be used when there is at least one active thermal unit or a storage unit without \ac{cf}.
Then the estimate $\rho(k)$ is based on the power set\hh point and the actual power output at time $k$.
This can be used to estimate the energy levels for the remaining storage units at $k + 1$.
\end{remark}

\subsection{Enhanced MPC problem with \ac{cf}}
In Problem~\ref{prob:simple:mpc},
the storage unit with \ac{cf} follows the default power set\hh point irrespectively of the renewable infeed or the load demand. This can lead to violation of the energy limits of the storage or can make the whole \ac{mg} unreliable. Therefore, we reformulate the previous \ac{mpc} for \ac{cf} such that the local control behaviour at the storage and thermal unit is taken into account. This allows to influence charging or discharging of the battery indirectly.

Let us define $\boldsymbol{\rho} = [\rho\ok ~ \cdots ~ \rho(k + h)]$, $h \in \N$.
The \ac{mpc} problem with \ac{cf} can be formulated as
\begin{problem}[Enhanced \ac{mpc} with \ac{cf}] \label{prob:simple:mpc:cf}
  \begin{subequations}
    \[
      \min_{\mathbf{u}, \boldsymbol{\delta}, \boldsymbol{\rho}} \ V(\mathbf{p}, \boldsymbol{\delta})
    \]
		where $V(\mathbf{p}, \boldsymbol{\delta})$ is same as in Problem~\ref{prob:simple:mpc},
    subject to \eqref{eq:unit:power}, \eqref{eq:powerBalance}, \eqref{eq:power:limits}, \eqref{eq:problem:forecast}, previous switch status $\delta_{\tth}(k - 1)$ and current energy level for storage units given as
		\begin{align}
		x_{\tst, i}\ok = \begin{cases}
		x_{\tst, i}\ok, \text{ if }\zeta_{s, i} = 1,\\
		\text{ estimated from \eqref{eq:estimation:storage}}, \text{ if }\zeta_{s, i} = 0
	\end{cases}
\end{align}

\end{subequations}

\end{problem}

\begin{remark}
	In Problem~\ref{prob:simple:mpc:cf}, the power set\hh points of the units are provided in a such a way that local control at the storage and thermal units achieves power balance. This is realised in the problem with $\boldsymbol{\rho}$ as a decision variable. Here we did not include any constraints on $\boldsymbol{\rho}$ that depend on the local control. However, one can include additional constraints on $\boldsymbol{\rho}$ easily.
\end{remark}
\begin{remark}
Any feasible power set\hh points for Problem \ref{prob:simple:mpc} are also feasible for Problem \ref{prob:simple:mpc:cf}. Therefore, for given current state and forecasts, the soultion of Problem \ref{prob:simple:mpc:cf} cannot be worse than the solution of Problem \eqref{prob:simple:mpc}.
\end{remark}

\subsection{Impact of communication failure}
In~\cite{FalFuWu2012}, it is mentioned that \ac{cf} could deteriorate the operation of the \ac{mg} both in terms of feasible region and operation cost. It is worthwhile to discuss this deterioration caused by \ac{cf} in Problem~\ref{prob:simple:mpc:cf}.

\begin{proposition}\label{pros:special:case}
	Consider the case with a single \ac{cf} in the system, occurring at a storage unit. Now, for given current energy levels and forecasts, the open\hh loop optimal solutions of Problem~\ref{prob:simple:mpc:cf} with and without the above mentioned \ac{cf} are equivalent. 
\end{proposition}
\begin{proof}
	It is sufficient to show that any power profile which is feasible without \ac{cf} is also feasible with \ac{cf} and vice\hh versa.
	Let us denote a feasible power profile without \ac{cf} as $\mathbf{p}^{1} = [p^{1}\ok ~ \cdots ~ p^{1}(k + h)]$.
	Let us assume that the storage unit with \ac{cf} is $q \in S$ and its corresponding default power set\hh points are $\mathbf{d}^{2}_{\tst, q} = [d_{\tst, q}^{2}\ok ~ \cdots ~ d_{\tst, q}^{2}(k + h)]  $.
	Now we need to show that $\mathbf{p}^{1}$ is feasible for the problem with \ac{cf}. This is possible only if we can select the following $\rho(k + j)$ and find suitable power set\hh points for the units with communication:
	\begin{equation}
		\rho(k + j) = \nicefrac{(p_{\tst, q}^{1}\ok - d_{\tst, q}^{2}\ok)}{\chi_{\tst, q}\ok}. \nonumber
	\end{equation}
	In the remaining units without \ac{cf}, there are no constraints on the power set\hh points. At any storage unit $q^{\prime} \in S\backslash q$, the power $p_{s,q^{\prime}}^{1}(k+j)$ can be generated by providing a set\hh point
	\begin{equation}
	 u_{\tst, q^{\prime}}^{2}(k + j) = p_{s,q^{\prime}}^{1}(k+j) - \chi_{\tst, q^{\prime}}\rho(k + j). \nonumber
 \end{equation}
 Similarly, we can show how the thermal and \ac{res} units can generate the power output $p^{1}_{q^{\prime}}(k + j), \forall q^{\prime} \in R \cup T$. Therefore, $\mathbf{p}^{1}$ is also a feasible for the problem with \ac{cf}.

 For proving the vice\hh versa condition, consider now  $\mathbf{p}^{1}$ as a feasible power profile with \ac{cf}. Now in the problem without \ac{cf}, picking $\rho(k +j) = 0$ and $\mathbf{u} =  \mathbf{p}^1$ would acheive this power profile. This concludes the proof that both cases are equivalent in terms of power profile and objective value.
\end{proof}


\section{Case Study}\label{sec:caseStudy}
This section presents a case study that analyses the performance of the \ac{mpc} with \acp{cf}.

\begin{figure}[t]
  \centering
  \includegraphics{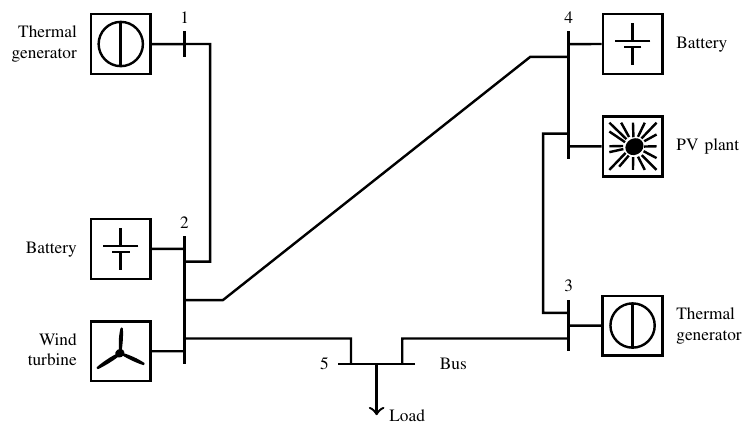}
  \caption{Microgrid considered in the case study.}
  \label{fig:oreoModel}
\end{figure}

\begin{table}[t]
  \centering
  \caption{Weights of cost function}
  \label{table:weights}
  \begin{tabular}{l l c l l}
  \toprule
  {Weight} & {Value} & & {Weight} & {Value}\\
  \cmidrule{1-2} \cmidrule{4-5}
  $a_{\tth}$   & $[0.43~~~ 0.67]\T                          $ & & $a_{\tst}$  & $[0.03~~~ 0.04]\T~ \rfrac{1}{\text{pu}^2}$       \\
  $a_{\tth,1}$ & $[0.335~~~ 0.475]\T                        $ & & $a_{\tst,1}$& $[50~~~ 50]\T~ \rfrac{1}{\text{puh}^2}$    \\
  $a_{\tth,2}$ & $[1.116~~~ 1.044]\T ~ \rfrac{1}{\text{pu}} $ & & $a_{\trs}$  & $[0.03~~~ 0.04]\T~ \rfrac{1}{\text{pu}^2}$    \\
  $a_{\tth,3}$ & $[1.685~~~ 0.778]\T~ \rfrac{1}{\text{pu}}$ & & &  \\
  \bottomrule
\end{tabular}
\end{table}

\begin{table}[t]
  \centering
  \caption{Operation limits and unit parameters of the \ac{mg}}
  \label{table:limits}
  \begin{tabular}{l l c l l}
  \toprule
  {Parameter} & {Value} & & {Parameter} & {Value}\\
  \cmidrule{1-2} \cmidrule{4-5}
  $[p^\text{min}_{\tth}~p^\text{max}_{\tth}]$ & $\begin{bmatrix} 0.08 & 0.6 \\ 0.17 & 1.0 \end{bmatrix} ~\text{pu}$ & & $[x^\text{min}_{}~x^\text{max}_{}]$ & $\begin{bmatrix} 0 & 2 \\ 0 & 3 \end{bmatrix}~ \text{puh}$ \\ \addlinespace[.3ex]
  $[p^\text{min}_{\tst}~p^\text{max}_{\tst}]$  & $\begin{bmatrix} -1 & 1.0 \\ -0.75 & 0.75 \end{bmatrix} ~\text{pu} $ & & $[x^\text{min}_{\tst}~x^\text{max}_{\tst}]$ & $\begin{bmatrix} 0.2 & 1.8 \\ 0.3 & 2.7 \end{bmatrix}~ \text{puh}$ \\ \addlinespace[.3ex]
  $[p^\text{min}_{\trs}~p^\text{max}_{\trs}]$ & $\begin{bmatrix} 0 & 2 \\ 0 & 2 \end{bmatrix} ~\text{pu} $ & & $\chi_{\tth}$  & $~[0.6~~~ 1]\T$ \\ \addlinespace[.3ex]
  $[p^\text{min}_\text{el}~p^\text{max}_\text{el}]$ & $\begin{bmatrix} -1.0 & 1.0 \end{bmatrix} ~\text{pu} $ & & $\chi_{\tst}$
  & $~[0.5 ~~~ 1]\T$ \\ \addlinespace[.3ex]
  $x^0$ & $~[1.3~~~ 0.8]\T ~\text{puh}$ & & $\eta_{\tst}$ & $~[0.95~~~ 0.95]\T$\\ \addlinespace[.3ex]
  $\delta^0_{\tth}$ & $~[0~~~ 0]\T$   & &  
  \\ \addlinespace[.3ex]
  \bottomrule
\end{tabular}
\end{table}

In the case study, we consider the \ac{mg} topology shown in Figure \ref{fig:oreoModel}.
It consists of two storage devices, two thermal generators, two \ac{res} and one load.
The power profile for available renewable infeed was provided by~\cite{ARM2009} and a load profile that emulates realistic behaviour was applied (see Fig. \ref{fig:powerProfiles}).
\begin{figure}[t]
  \centering
  \includegraphics{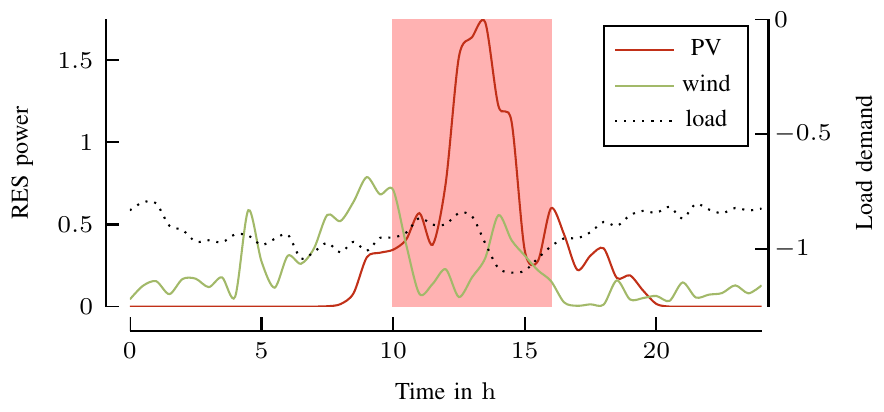}
  \caption{Power profiles of \ac{res} and the load demand. The highlighted part shows the \ac{cf} period between $10.00$ to $16.00$.}
  \label{fig:powerProfiles}
\end{figure}
The generator parameters were taken from \cite{ParRikGli2014} and transformed to per-unit (\unit{pu}).
The weights of the operational cost and the operational limits can be found in Tables \ref{table:weights} and \ref{table:limits}. In \ac{mpc}, the sampling time is $T_s = \unit[30]{min}$, the prediction horizon is $h = 12$ and the discount factor is $\gamma = 0.95$.
A naive forecaster provides the forecasts for the \ac{res} infeed and the load demand of the \ac{mg}~\cite{hyndman2018}.
Finally, the \ac{mpc} problem is implemented in Matlab\textsuperscript{\textregistered}~2017b with YALMIP~\cite{Lof2004} and solved with Gurobi 6.5.
The analysis was carried out for a simulation period of one day.
First, the closed\hh loop simulation for this period is performed with \ac{mpc} without occurrence of any \ac{cf}. The results are labelled as reference in Figures
\ref{fig:simulationResultsScen-1} and
\ref{fig:simulationResultsScen-2} and Table \ref{table:simulationResults}. These results serve as reference for the evaluation of the scenarios with \acp{cf}.

\subsubsection*{Simulation with \ac{cf}}
The duration of the \ac{cf} is selected as $\unit[6]{h}$, from $10.00$ to $16.00$, which is highlighted in Figure~\ref{fig:powerProfiles}.
It can be observed that \ac{res} have high infeed during this period and storage units ideally store the excess generation. Two \acl{cf} scenarios are considered in this period:
\begin{enumerate}
   \item[-] \textit{Scenario} I : \ac{cf} at one storage unit - battery at bus $4$.
   \item[-] \textit{Scenario} II : \ac{cf} at both the storage units.
\end{enumerate}
Closed\hh loop simulations are carried out for these scenarios with 1) \ac{mpc} Problem~\ref{prob:simple:mpc}, which is referred to as \ac{mpc} and 2) \ac{mpc} Problem~\ref{prob:simple:mpc:cf} which is referred to as enhanced \ac{mpc}.
The energy of the batteries in scenario I and scenario II is shown in Figure~\ref{fig:simulationResultsScen-1} and Figure~\ref{fig:simulationResultsScen-2}.
The energy wasted from the \ac{res} and energy output from the thermal units during the closed\hh loop simulation for the above scenarios
are summarised in Table~\ref{table:simulationResults}.

\begin{figure}[t]
  \centering
  \includegraphics{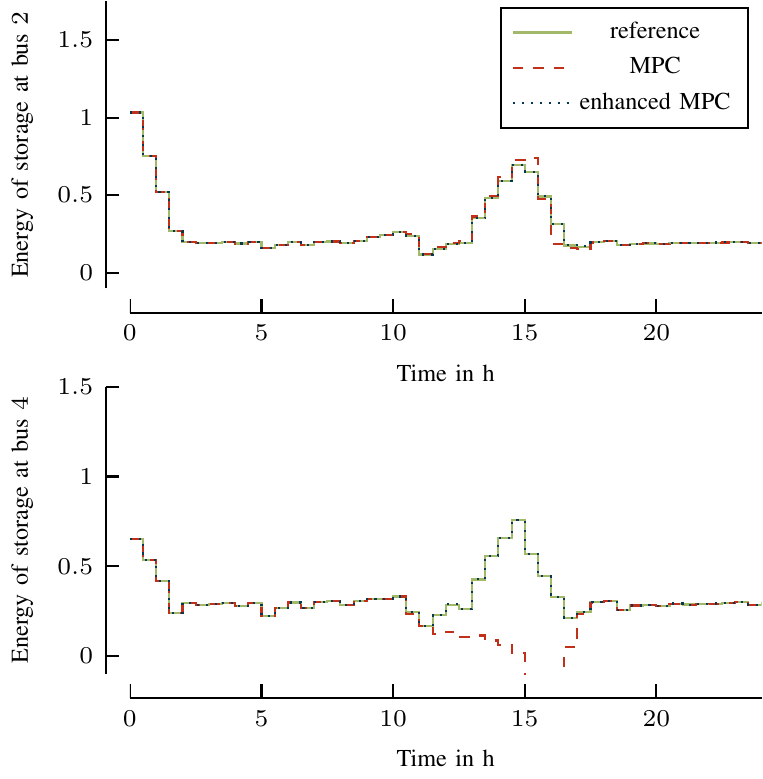}
  \caption{Closed-loop simulations in scenario I: Energy levels of storages with 1) reference 2) \ac{mpc} 3) enhanced \ac{mpc}.}
  \label{fig:simulationResultsScen-1}
\end{figure}

\begin{figure}[t]
  \centering
  \includegraphics{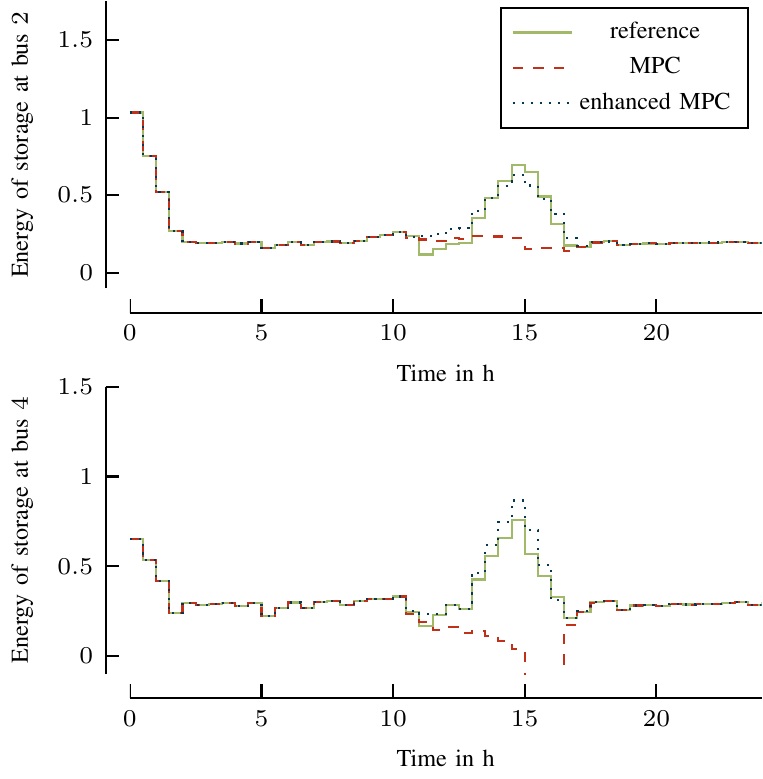}
    \caption{Closed-loop simulations in scenario II: Energy levels of storages with 1) reference 2) \ac{mpc} 3) enhanced \ac{mpc}.}
  \label{fig:simulationResultsScen-2}
\end{figure}

In scenario I, the energy of the batteries with enhanced \ac{mpc} coincides with the reference \ac{mpc} (without \ac{cf}). This is  empirical evidence that the statement made on open\hh loop behaviour in Proposition~\ref{pros:special:case} also applies to closed\hh loop behaviour.
In the case with \ac{mpc}, the storage at bus $4$ cannot be influenced during the \ac{cf} period and only the storage at bus $2$ is utilised for storing energy. This effect can be noticed in the storage at bus $2$ whose energy level is higher than reference for some periods during the \ac{cf}.
Furthermore, from Table~\ref{table:simulationResults} it can be observed that the wasted \ac{res} energy and the thermal energy with both the reference and the enhanced \ac{mpc} coincide. \ac{mpc} leads to an increase in \ac{res} wastage that is compensated by the thermal units.

In scenario II, the energy of the batteries with enhanced \ac{mpc} slightly differs from the reference. However, it can be noticed that the energies are close to the reference profile. In the case with \ac{mpc} both the storages cannot be influenced and therefore their energy profiles deviate from the reference profile.
Table~\ref{table:simulationResults} shows that the enhanced \ac{mpc} results in a slightly increased wastage of \ac{res} energy and the thermal energy compared to scenario I.
However, these values are lower than the \ac{mpc}.

\begin{table}[t]
  \centering
  \caption{\ac{res} energy wastage and the thermal energy output in closed\hh loop simulation}
  \label{table:simulationResults}
  \begin{tabular}{c c c c c c}
  \toprule
  & {without \ac{cf}} & \multicolumn{2}{c}{scenario I} & \multicolumn{2}{c}{scenario II} \\
   & {reference} & {\ac{mpc}} & \makecell{enhanced \\ \ac{mpc}} & {\ac{mpc}} & \makecell{enhanced \\ \ac{mpc}}\\
  \midrule
  \makecell{\ac{res} wastage \\ energy ($\unit{puh}$)} & 2.68 & 3.16 & 2.68
    & 3.85 & 2.73\\
  \makecell{thermal energy\\ ($\unit{puh}$)} & 11.68 & 12.10 & 11.68
    & 12.77 & 11.71\\
  \bottomrule
\end{tabular}
\end{table}

We can conclude that in scenario I, where the \ac{cf} occurs at only one storage unit, the enhanced \ac{mpc} results in the same power profiles as the reference.
In both scenario I and scenario II, the enhanced \ac{mpc} has a better performance than the \ac{mpc}.


\section{Conclusion}\label{sec:conclusion}
In this paper, we proposed an \ac{mpc} formulation for the operation of an \ac{mg} during communication failure.
This formulation takes into account the hierarchical control layers of the \ac{mg}. This allows to adjust the power output of a unit affected by communication failure.
In particular, we formulated the problem when the communication failure occurs in storage units. In this case, we also discuss the estimation of the current energy level of the battery with communication failure. Finally, we show the benefits of the proposed approach with a numerical example.

The proposed approach assumes only communication failures and therefore a possible extension is to include electrical failures. Also, the current approach does not account for uncertainties in the forecasts and possible future communication failures. Including them in a robust framework is another future direction.

\bibliographystyle{IEEEtran}
\bibliography{IEEEabrv,literature}


\end{document}